\documentclass[aps,prl,twocolumn,superscriptaddress,floatfix,citeautoscript,preprintnumbers]{revtex4-2}
\usepackage{amsfonts, amssymb, amsmath, amsthm}
\usepackage{array}
\usepackage{bm}
\usepackage{braket}
\usepackage{dcolumn}
\usepackage{multirow}
\usepackage{makecell}
\usepackage{graphicx}
\usepackage{hyperref}
\hypersetup{
  colorlinks,
  citecolor=blue,
  filecolor=blue,
  linkcolor=blue,
  urlcolor=blue
}
\usepackage[normalem]{ulem}
\usepackage{dsfont}
\usepackage{indentfirst}
\usepackage{listings}
\usepackage{subfigure}
\usepackage{float}
\usepackage{titlesec}
\usepackage{color}
\usepackage{xcolor}
\usepackage{enumitem}

\newtheorem{definition}{Definition}

\newtheorem{lemma}{Lemma}

\newcommand{\dInt}{{\mathrm{d}}}
\newcommand{\id}{{\mathds{1}}}
\newcommand{\tr}{{\mathrm{Tr}}}
\newcommand{\re}{{\mathrm{Re}}}
\newcommand{\im}{{\mathrm{Im}}}

\newcommand{\varO}{\mathcal{O}} 
\newcommand{\pd}[1]{\frac{\partial}{\partial #1}}

\begin{document}

\title{Phase-Sensitive Quantum Measurement without Controlled Operations}

\begin{abstract}

Many quantum algorithms rely on the measurement of complex quantum amplitudes. Standard approaches to obtain the phase information, such as the Hadamard test, give rise to large overheads due to the need for global controlled-unitary operations. 
We introduce a quantum algorithm based on complex analysis that overcomes this problem for amplitudes that are a continuous function of time.
Our method only requires the implementation of real-time evolution and a shallow circuit that approximates a short imaginary-time evolution. We show that the method outperforms the Hadamard test in terms of circuit depth and that it is suitable for current noisy quantum computers when combined with a simple error-mitigation strategy.
\end{abstract}

\author{Yilun Yang}
\affiliation{Max-Planck-Institut f\"ur Quantenoptik, Hans-Kopfermann-Str.\ 1, D-85748 Garching, Germany}
\affiliation{Munich Center for Quantum Science and Technology (MCQST), Schellingstr. 4, D-80799 M\"unchen}
\author{Arthur Christianen}
\affiliation{Max-Planck-Institut f\"ur Quantenoptik, Hans-Kopfermann-Str.\ 1, D-85748 Garching, Germany}
\affiliation{Munich Center for Quantum Science and Technology (MCQST), Schellingstr. 4, D-80799 M\"unchen}
\author{Mari Carmen Ba\~nuls}
\affiliation{Max-Planck-Institut f\"ur Quantenoptik, Hans-Kopfermann-Str.\ 1, D-85748 Garching, Germany}
\affiliation{Munich Center for Quantum Science and Technology (MCQST), Schellingstr. 4, D-80799 M\"unchen}	
\author{Dominik S. Wild}
\affiliation{Max-Planck-Institut f\"ur Quantenoptik, Hans-Kopfermann-Str.\ 1, D-85748 Garching, Germany}
\affiliation{Munich Center for Quantum Science and Technology (MCQST), Schellingstr. 4, D-80799 M\"unchen}
\author{J. Ignacio Cirac}
\affiliation{Max-Planck-Institut f\"ur Quantenoptik, Hans-Kopfermann-Str.\ 1, D-85748 Garching, Germany}
\affiliation{Munich Center for Quantum Science and Technology (MCQST), Schellingstr. 4, D-80799 M\"unchen}	

\date{\today}							
\maketitle

\paragraph{{Introduction.---}}

The complex phases of quantum amplitudes play an essential role in quantum algorithms~\cite{Shor1994, Kitaev1995, Knill2007, Harrow2009, Wiebe2016, Clinton2021} and quantum sensing~\cite{Degen2017}. Many algorithms require measuring the relative phase between two quantum states~\cite{Aharonov2006, OBrien2019, Somma2019, McArdle2020, Lu2021, Lin2022, Huggins2022, Ding2023, Patti2023, Ni2023}. A common subroutine for this purpose is the Hadamard test, which converts phase information into probabilities by means of interference~\cite{Cleve1998}. Despite impressive experimental progress, 
the Hadamard test remains out of reach for most applications owing to the challenge of implementing the required controlled-unitary operation. In this Letter, we propose an alternative method to determine the complex overlap between certain states that uses no ancillary qubits or global controlled-unitary operations. Unlike other ancilla-free schemes~\cite{Russo2021, Lu2021}, our approach does not require the preparation of superpositions with a reference state, which are highly susceptible to noise~\cite{Laflamme1998a, Leibfried2005, Song2017, Omran2019, Wei2020, Bin2022}. Instead of interference, our method hinges on the principles of complex analysis.

\begin{figure}
    \centering
    \includegraphics[width=.45\textwidth]{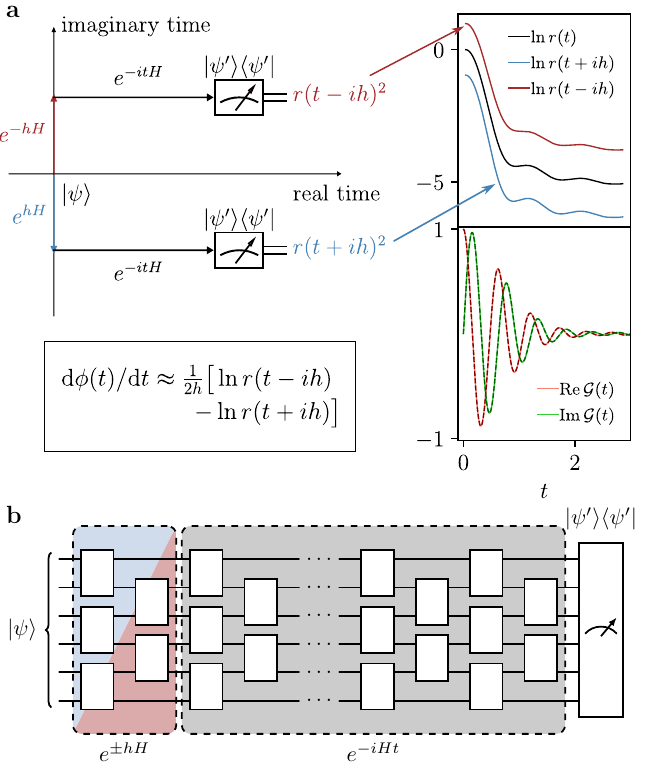}
    \caption{\textbf{(a)} The time derivative of the complex phase $\phi(t)$ of Loschmidt amplitude $\mathcal{G}(t)$ can be estimated from $r(t \pm ih) = | \braket{\psi' | e^{-i H t} e^{\pm h H} | \psi} |$. The right panel shows the result of our approach with $h = 0.1$ for the transverse-field Ising chain, Eq.~(\ref{eq:ising}), of length $N=40$. The solid lines on the lower right correspond to the complex Loschmidt amplitude obtained from the algorithm, while the overlapping dashed lines indicate the exact result. 
    \textbf{(b)} Circuit to measure $r(t \pm ih)$. For initial product states, the rescaled imaginary-time evolution has the same brickwork layout as a single real-time Trotter step.}
    \label{fig:scheme}
\end{figure}

The proposed approach applies to overlaps of the form of the (generalized) Loschmidt amplitude
\begin{eqnarray}
    \mathcal{G}(t) = \braket{\psi^{\prime} | e^{-iHt} |\psi},
    \label{eq:loschmidt}
\end{eqnarray}
where $H$ is a local Hamiltonian. Our algorithm requires that $\ket{\psi}$ has a short correlation length and that $\ket{\psi'}$ can be prepared using a unitary circuit. These assumptions are needed to be able to efficiently apply a short imaginary-time evolution to $\ket{\psi}$~\cite{Jones2019, Motta2020, Lin2021, Nishi2021, Jouzdani2022, Kamakari2022, Gilyen2019, Holmes2022, An2023} and to perform a projective measurement onto the final state $\ket{\psi'}$~\footnote{The measurement is defined by a complete, mutually orthogonal set of projection operators that includes $\ket{\psi'}\bra{\psi'}$. We assign $1$ to an outcome of $\ket{\psi'}$ and $0$ all other outcomes. For any state that can be efficiently prepared from a product state in the computational basis using unitary circuits, one can invert the preparation circuit and perform local measurements in the computational basis to realize the projective measurement. Such states include time evolved product states and injective matrix product states~\cite{malz2023}.}. The absolute value $|\mathcal{G}(t)|$ can be obtained by repeatedly evolving $\ket{\psi}$ and averaging over projective measurements onto $\ket{\psi'}$. Here we describe how to obtain the phase.

Equation~(\ref{eq:loschmidt}) includes several cases of interest. When $\ket{\psi'} = \ket{\psi}$, $\mathcal{G}(t)$ is the Fourier transform of the local density of states, which has applications in the study of quantum chaos~\cite{Benenti2002, Andersen2006}, in optimal measurements of multiple expectation values~\cite{Huggins2022}, and in estimating energy eigenvalues~\cite{Somma2019, OBrien2019, Lin2022, Ding2023, Ni2023}. The case when $\ket{\psi'} = A e^{-i H t'} \ket{\psi}$, for a local unitary $A$, 
is relevant for probing thermal properties of many-body systems~\cite{Lu2021,Schuckert2023,Hemery2023,Ghanem2023}.

The key idea underlying our method is to view $\mathcal{G}$ as a function of a complex variable $z$. Assuming that $\mathcal{G}(z)$ is analytic and nonzero, the Cauchy-Riemann equations imply that the real-time derivative of the phase of $\mathcal{G}(z)$ is equal to the derivative of $\ln |\mathcal{G}(z)|$ along the imaginary-time direction. We use this relation to obtain the desired phase by carrying out the following three steps on a quantum computer (see Fig.~\ref{fig:scheme}). First, a quantum circuit applies an evolution under the Hamiltonian $H$ for a short \emph{imaginary time} $h$ to the initial state $\ket{\psi}$~\cite{Jones2019, Motta2020, Lin2021, Nishi2021, Jouzdani2022, Kamakari2022}. Second, we evolve the system under $H$ for the \emph{real time} $t$. Third, we perform a projective measurement onto the state $\ket{\psi'}$ by inverting the circuit that prepares $\ket{\psi'}$ from a computational basis state, followed by measurements in the computational basis. Using these steps, we can estimate $|\mathcal{G}(t \pm i h)|$. This yields a finite-difference approximation to the imaginary-time derivative of $\ln |\mathcal{G}(z)|$, which is equal to the real-time derivative of the phase. We finally compute the phase of the Loschmidt amplitude by repeating these steps for different values of $t$ and numerically integrating the derivative, starting from a time at which the phase is known.

We show below that our method is efficient if $|\mathcal{G}(t \pm i h)|$ is bounded from below by an inverse polynomial in the system size $N$. For product states $\ket{\psi'} = \ket{\psi}$, however, the Loschmidt amplitude decays as a Gaussian function on a time scale $\mathcal{O}(1/\sqrt{N})$~\cite{Hartmann2004}. In this case, our approach will be inefficient even for short constant times, for which the Loschmidt amplitude can be computed by a polynomial-time, classical algorithm~\cite{Wild2023}. By contrast, no efficient classical algorithm is known for the case $\ket{\psi'} = A e^{- i H t'} \ket{\psi}$. Since the real- and imaginary-time evolution operators commute, our method can be used to compute the phase as a function of $t - t'$, with the expectation value $\braket{\psi | e^{i H t^{\prime}} A e^{- i H t^{\prime}} | \psi}$ serving as the reference for the integration. This is expected to be classically hard even for small $t - t'$ since computing the reference value at times $t' = \mathrm{poly}(N)$ is BQP-complete~\cite{Janzing2005}.

Although our approach is based on the analytic properties of a function of a continuous variable, we show below that it also works well in the discrete setting of Trotter evolution. Hence, the method applies to both circuit-based quantum computers and to analog quantum simulators supplemented by shallow circuits to implement the imaginary-time evolution. To demonstrate the suitability of the method for near-term quantum devices, we combine it with a simple error-mitigation strategy~\cite{Temme2017, Endo2018, Kandala2019, OBrien2021, Cai2022, Yang2023} and show numerically that the phase can be reliably recovered in a system of $N=24$ qubits. Beyond providing a viable alternative to the Hadamard test on near-term quantum computers, our method may be useful in the early fault-tolerant regime as the absence of controlled global operations significantly reduces the circuit depth.

\paragraph{{Theoretical approach.---}} 
To formally describe the algorithm, we consider the complex variable $z = t - i \beta$, where $t$ represents real time and $\beta$ stands for imaginary time or inverse temperature. The generalized Loschmidt amplitude, Eq.~(\ref{eq:loschmidt}), can be decomposed into its absolute value and phase according to
\begin{eqnarray}
    \mathcal{G}(z) = r(z) e^{i\phi(z)},
\end{eqnarray}
where $0 \le r(z) \le 1$ and $\phi(z)$ is real. In a system of finite size, the Loschmidt amplitude can be written as a sum of exponentials by expanding the states $\ket{\psi}$ and $\ket{\psi'}$ in the energy eigenbasis. The logarithm $\ln \mathcal{G}(z)$ is therefore holomorphic everywhere except when $\mathcal{G}(z) = 0$. For an  analytic branch of $\phi(z)$, the Cauchy-Riemann equations applied to $\ln \mathcal{G}(z) = \ln r (z) + i \phi(z)$ give
\begin{eqnarray}
    \pd{t} \phi(z) = \pd{\beta} [\ln r(z)].
    \label{eq:der}
\end{eqnarray}
Therefore, if $\mathcal{G}(t) \neq 0$ in the interval $[t_1, t_2]$, the phase difference $\phi(t_2) - \phi(t_1)$ can be computed as
\begin{eqnarray} \label{eq:int_phase}
    \phi(t_2) - \phi(t_1) = \int_{t_1}^{t_2} \pd{\beta} \left[\ln r(z)\right]_{|\beta=0} \dInt t.
\end{eqnarray}
If the phase $\phi(t_1)$ is known, then $\phi(t_2)$ may be computed from the partial derivative of $r(z)$ along the imaginary-time direction. In practice, we numerically approximate the partial derivative by the mid-point formula
\begin{eqnarray}
    \pd{\beta} [\ln r(z)]_{|\beta = 0} \approx \frac{\ln r(t - i h) - \ln r(t + i h)}{2h},
    \label{eq:der_approx}
\end{eqnarray}
where $h$ is a small parameter. 

This procedure is well defined for $r(t \pm i h) > 0$ in the interval $[t_1, t_2]$. To bound the computational errors, we make the stronger assumption that $|\ln r(z)| \le c N$ at all points in the complex plane within distance $a$ of the interval $[t_1, t_2]$, for constants $c$ and $a > h$. 
In the case of Trotter evolution, we make analogous assumptions for closely related functions~\cite{supp_mat}.
We highlight, however, that our approach can be extended to treat zeros in $\mathcal{G}(t)$ by separately considering the resulting discontinuities in the phase~\cite{supp_mat}.

The above analysis has reduced the problem to measuring the absolute values $r(t \pm i h)^2$. It involves nonunitary imaginary-time evolution which cannot be directly applied. However, Motta \textit{et al.}~\cite{Motta2020} showed that $e^{\pm h H}$ can be simulated by a unitary circuit for short times $h$ if the spatial correlations of $\ket{\psi}$ decay exponentially with correlation length $\xi$ and $H$ is a local Hamiltonian. For each local term $H_m$ in the Hamiltonian, it is possible to approximate $e^{\pm h {H_m}}\ket{\psi} \approx c_m^{\pm} V_m^\pm \ket{\psi}$, where $V_m^\pm$ is a local unitary and $c_m^{\pm} = \sqrt{\braket{\psi | e^{\pm 2hH_m} |\psi}}$ accounts for the normalization. Since $h$ is small, the product over all $V_m^\pm / c_m^\pm$ (in arbitrary order) is a good approximation of $e^{-h H}$. The operators $V_m^\pm$ act on $\varO((\xi \log N)^d)$ qubits in $d$ spatial dimensions and the complexity of computing $V_m^\pm$ is quasi-polynomial in $N$ for $d>1$. This renders the approach challenging for large $\xi$. Below, we focus on the simplest case of product states, for which the unitaries $V_m^\pm$ act on the same sites as $H_m$ and can be efficiently computed. The resulting circuit has the same structure as a single Trotter step~\cite{supp_mat}.

\paragraph{{Error analysis.---}}
We next analyze the error in the estimated phase arising from the different approximations in our algorithm.  The approximation error of the imaginary-time evolution is dominated by the first-order Trotter decomposition, which results in the phase error~\cite{supp_mat}
\begin{eqnarray}
    \Delta{\phi}_{\mathrm{ITE}} =\varO(Nth^2).
    \label{eq:err_qite}
\end{eqnarray}
The factor $t = t_2-t_1$ accounts for the accumulation of errors in the integral in Eq.~\eqref{eq:int_phase}. While the real-time evolution can be carried out exactly on analog quantum simulators, digital quantum computers incur an additional Trotter error, leading to the phase error~\cite{supp_mat}
\begin{eqnarray}
    \Delta{\phi}_{\mathrm{RTE}} = \varO(N t^2\tau^p).
    \label{eq:err_rte}
\end{eqnarray}
Here, $\tau$ is the time of a single Trotter step, $p$ is the order of the Trotter decomposition~\cite{Childs2021}, and we again included the accumulation of errors in Eq.~\eqref{eq:int_phase}. Numerical differentiation and integration give rise to additional errors. They can, however, be safely ignored for practical orders of the Trotter expansion ($p \le 4$) as they are asymptotically at most as big as $\Delta{\phi}_{\mathrm{ITE}}$ and $\Delta{\phi}_{\mathrm{RTE}}$~\cite{supp_mat}.

\begin{figure}
    \centering
    \includegraphics[width=.45\textwidth]{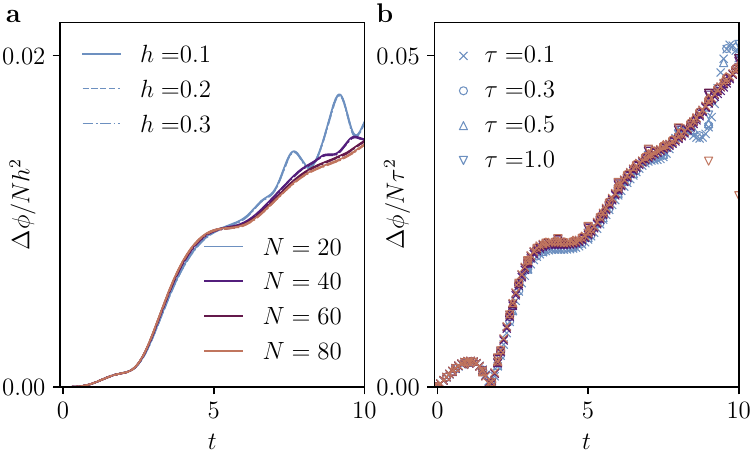}
    \caption{Error in the phase of Loschmidt amplitude, $\Delta \phi$, computed using our approach for the transverse-field Ising chain, Eq.~(\ref{eq:ising}). \textbf{(a)}~$\Delta \phi / N h^2$ as a function of time $t$ with fixed real-time Trotter step $\tau = 0.01$ for different values of the imaginary-time step $h$ and different system sizes $N$. \textbf{(b)}~$\Delta \phi / N \tau^2$ for $h = 0.01$ and different values of $\tau$ and $N$. The color coding is the same as in (a).
    }
    \label{fig:trotter} 
\end{figure}

We verify these analytic estimates using numerical results for the transverse-field Ising chain, whose Hamiltonian is given by
\begin{eqnarray}
    H = - J \sum_{i=1}^{N-1}S^z_i S^z_{i+1} + g\sum_{i=1}^{N} S^x_i.
    \label{eq:ising}
\end{eqnarray}
Throughout this work, we set $J = 1$ and $g = 0.5$, corresponding to the ferromagnetic phase. Both states $\ket{\psi} $ and $\ket{\psi^{\prime}}$ are chosen as $\ket{\uparrow\uparrow\uparrow\cdots}$. For the Trotter decomposition, we alternate between the ferromagnetic and transverse field terms.

Figure~\ref{fig:trotter} shows the error in the phase computed using our approach. The numerical results were obtained by matrix product state simulations with bond dimension 200, for which truncation errors are negligible~\cite{supp_mat}. In Fig.~\ref{fig:trotter}(a), we set $\tau = 0.01$ such that the error in the imaginary-time evolution dominates. The phase error collapses onto a single curve upon dividing by $N h^2$, which confirms the predicted error due to imaginary-time evolution, Eq.~(\ref{eq:err_qite}). Similarly, we set $h = 0.01$ in Fig.~\ref{fig:trotter}(b) to isolate the effect of the real-time Trotter error. The collapse of the data agrees with  Eq.~(\ref{eq:err_rte}).

In addition to numerical errors, any experiment incurs statistical errors. Given $M$ measurements, a probability $p$ estimated by counting successful outcomes will have a multiplicative error $\sqrt{( 1 - p) / Mp}$, governed by the standard deviation of the binomial distribution. According to Eq.~(\ref{eq:der_approx}), for the measured probabilities $p_{\pm}(t) = r(t \pm i h)^2 / \prod_m (c_{m}^{\pm})^2$, this contributes an additive error 
\begin{equation}
    \Delta{\phi}_\mathrm{S} = \varO\left( \frac{It}{h \sqrt{M}} \right)
\end{equation}
to the final phase for $M$ measurements per time step. The integral in Eq.~(\ref{eq:int_phase}) is included in the factor $I = \int_{t_1}^{t_2} \dInt t' \left[ \sqrt{1/{ p_+(t') }} + \sqrt{1/{ p_{-}(t') }} \right] / t$. In contrast to the previous errors, the statistical error depends on the magnitude of the measured probabilities. 

\begin{figure*}
    \centering
    \includegraphics[width=.95\textwidth,trim={0.2cm 0.3cm 0.2cm 0.2cm},clip]{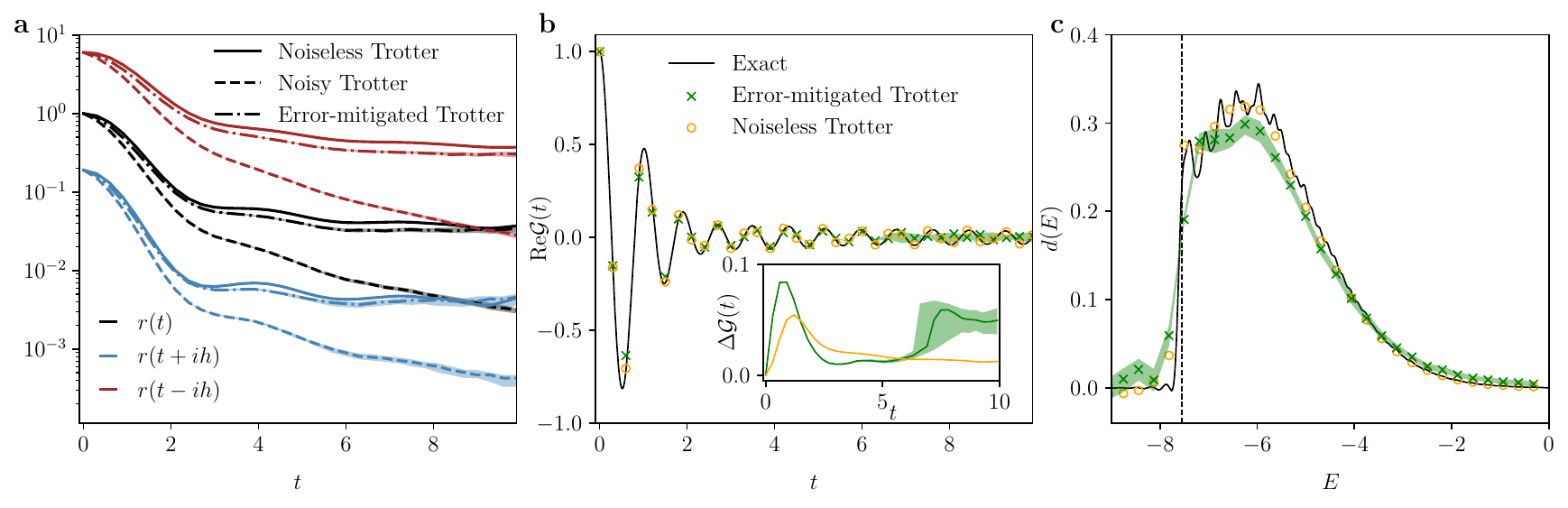}
    \caption{\textbf{(a)}~Absolute value of the Loschmidt amplitudes for an Ising chain of length $N=24$ with initial state $\ket{\uparrow \uparrow \uparrow \cdots}$ and Trotter step sizes $\tau = h = 0.3$. The dashed lines include single-qubit depolarizing noise with probability $\gamma= 3 \times 10^{-3}$ after each gate. The dash-dotted lines are obtained by the error mitigation described in the text. We quantify the statistical error of the error-mitigated curves by simulating $100$ experiments, each of which uses $M = 10^6$ measurements to estimate the survival probability. The dash-dotted line corresponds to the median of the $100$ experiments, while the shaded areas indicate the range between the first and third quartile.
    \textbf{(b)}~Real part of the Loschmidt amplitude computed from the data in (a) using our algorithm. The exact value under continuous time evolution is plotted for reference. The inset shows the absolute difference of the reconstructed values from the exact amplitude.
    \textbf{(c)}~The LDOS obtained through discrete Fourier transform from the data in (b). The vertical, dashed line indicates the exact ground state energy $E_0\approx-7.55$.}
    \label{fig:noisy}
\end{figure*}

\paragraph{{Comparison with existing methods.---}}

\begin{table}
    \centering
    \begin{tabular}{l l l}
        \hline 
        \hline
        Method & $D$ & $M$\\
        \hline
        Hadamard test & $\varO( t^{1+\frac{1}{p}} N^{1 + \frac{1}{d}+\frac{1}{p}} / \epsilon^{\frac{1}{p}} )$ \quad & 
        $\varO(  1 / \epsilon^2)$\\[0.3em]
        \makecell[l]{Sequential\\interferometry} & $\varO(r^{\frac{1}{p}} t^{1+\frac{1}{p}} N^{\frac{2}{p}}  / \epsilon^{\frac{1}{p}})$  & $\varO(\tilde{I}^2 r^2 N^2 / \epsilon^2 )$\\[0.3em]
        This work & $\varO( 
        r^{\frac{1}{p}}t^{1+\frac{2}{p}} N^{\frac{1}{p}} / \epsilon^{\frac{1}{p}} )$ &
        $\varO \left( I^2 r^3 t^3 N / \epsilon^3 \right)$\\
        \hline 
        \hline
    \end{tabular}
    \caption{Circuit depth $D$ and number of measurements $M$ to estimate the complex Loschmidt amplitude $\mathcal{G}$ with additive error $\epsilon$. All protocols use a real-time Trotter decomposition of order $p$. The Hadamard test is implemented using a single ancilla qubit with swap operations in $d$ spatial dimensions.  
    The latter two methods require $M$ measurements at each intermediate state or time step, but the corresponding values of the phase are also returned. For these approaches, we only consider initial product states and $\epsilon$ bounds the error $r \Delta \phi$ arising from the uncertainty in the phase. The quantities $\tilde{I}$ and $I$ depend on the intermediate amplitudes in these sequences, see text and supplemental material~\cite{supp_mat}.
    }
    \label{tab:comp}
\end{table}

To compare our approach to existing methods, we consider the error $\Delta \mathcal{G}$ in the complex Loschmidt amplitude $\mathcal{G}$. This error is related to the phase error, $\Delta \phi$, by $|\Delta \mathcal{G}|^2 = \Delta r^2 + (r \Delta \phi)^2$. Here, $\Delta r$ is the error from an independent measurement of $r$, which only requires the Trotterized circuit without imaginary-time evolution. To bound $\Delta r$ by $\epsilon$, we need a circuit of depth $D_r = \varO(t / \tau) = \varO(  t^{1+\frac{1}{p}} N^{\frac{1}{p}} / \epsilon^{\frac{1}{p}} )$ and a number of $M_r = \varO \left( 1 /\epsilon^2 \right)$ measurements~\cite{supp_mat}. For the term $r \Delta \phi$, we bound the individual contributions to the phase error. For instance, $r \Delta \phi_\mathrm{ITE} < \epsilon$ implies that $h = \varO (\sqrt{ \epsilon / r N t } )$. A similar bound on the real-time evolution gives $\tau = \varO ( ( \epsilon / r N t^2 )^\frac{1}{p} )$, resulting in the circuit depth $D = \varO( r^\frac{1}{p} t^{1+\frac{2}{p}} N^{\frac{1}{p}} / \epsilon^{\frac{1}{p}} )$. Bounding the statistical error yields the number of measurements $M = \varO ( I^2 r^3 t^3 N  / \epsilon^3)$ for each time step. We note that when $r$ is bounded from below by a constant, the cost of estimating $\phi$ dominates.

We compare this resource cost to the Hadamard test and sequential interferometry~\cite{Lu2021}. The latter method employs a reference state whose Loschmidt amplitude, including the phase, is known. The details of these two methods are described in the supplemental material~\cite{supp_mat}. Table~\ref{tab:comp} summarizes the resource cost for each method. For a constant evolution time $t$, the circuit depth needed for our algorithm is reduced by a factor $\varO(N^{1 + 1/d})$ compared to the Hadamard test with swaps, and by $\varO(N^{1/p})$ compared to sequential interferometry. This improvement is particularly significant for noisy quantum computers, for which circuit depth is the key limiting factor.

\paragraph{{Applications.---}}

For practical applications of our protocol, it is important to consider the role of noise. We propose a simple rescaling strategy based on previous work to mitigate the effects of noise~\cite{Yang2023}. The key observation is that errors are unlikely to drive the system towards the target state $\ket{\psi'}$. Hence, the measured probabilities are decreased in a consistent fashion, which can be mitigated by rescaling with the probability of having no noise. This is equivalent to zero-noise extrapolation with an exponential fitting function~\cite{Li2017,Temme2017,Kim2023}. 
Below, we simply use the known noise rate for rescaling. In practice, the rescaling factor can be determined by enhancing the noise or by measuring the survival probability after forward plus backward evolution~\cite{Yang2023}.

As a proof of concept, we apply our approach to compute the local density of states (LDOS) $d(E)$ through the Fourier transform
\begin{eqnarray}
    d(E) = \braket{\psi | \delta(E - H) | \psi} = \frac{1}{2 \pi} \int_{-\infty}^\infty \mathcal{G}(t) e^{i E t} \, \mathrm{d} t.
\end{eqnarray}
If the initial state has a sufficiently large overlap with the ground state, its LDOS enables determining the ground-state energy.
In quantum chemistry, this can hold even for product states, rendering our approach particularly suitable \cite{Bauer2020}.
We further note that our approach is compatible with recent proposals that classically process the Loschmidt amplitudes at different times in order to solve the general quantum eigenvalue estimation problem~\cite{Somma2019, supp_mat} with Heisenberg-limited scaling~\cite{Lin2022, Ding2023, Ni2023}. 

We apply our approach to compute the LDOS of $\ket{\psi}=\ket{\uparrow \uparrow \uparrow \cdots}$ in an Ising chain of system size $N=24$. We numerically carry out the Trotter evolution with Trotter steps $\tau = h = 0.3$ using the Cirq library~\cite{cirq}. We add single-qubit depolarizing noise of rate $\gamma=3\times 10^{-3}$ after each layer of the quantum circuit. We average over 5000 trajectories of a Monte Carlo wavefunction simulation to obtain the probabilities $p_{\pm}$. The results are shown in Fig.~\ref{fig:noisy}. Here we have included statistical noise by simulating the experimental sampling procedure (see caption).

Figure~\ref{fig:noisy}(a) shows that the depolarizing noise is mitigated well by rescaling $r^2(t)$ and $r^2(t\pm ih)$ by $(1-\gamma)^{ND}$. The error in the reconstructed Loschmidt amplitude remains small within the range of $t$ in Fig.~\ref{fig:noisy}(b). We estimate the LDOS of the initial state by a discrete Fourier transform of the data in Figure \ref{fig:noisy}(b) and similar data for the imaginary part of $\mathcal{G}(t)$. The energy resolution is $\pi / t_{\max} \approx 0.31$, determined by the maximum time $t_{\max} = 10$. We show the result in Fig.~\ref{fig:noisy}(c) for both noisy, error-mitigated (green) and noiseless (orange) Trotter simulations. For reference, we also include the exact result (black line), which is broadened by a Gaussian of width $0.08$. For both Trotter simulations, the first point with $d(E)>0.1$ appears at $E\approx -7.50$, while the exact ground state energy is $E_0\approx-7.55$.

\paragraph{{Summary and outlook.---}}

We propose a quantum algorithm to estimate the phase of Loschmidt amplitudes applicable to states with short-ranged correlations. It can replace and outperform the Hadamard test for amplitudes that arise from continuous time evolution under a local Hamiltonian. While our analysis focused on generalized Loschmidt amplitudes, the approach can be readily extended to multiple time-evolution operators~\cite{supp_mat}, which renders it applicable to many quantities of physical interest including transport coefficients \cite{Agarwal2015, Kanasz2017, Parker2019} and out-of-time-ordered correlators (OTOCs)~\cite{Hashimoto2017, Sajjan2023}. The algorithm requires no ancillary qubits or controlled operations. When combined with a simple error-mitigation strategy, our algorithm may enable phase-sensitive measurements on current noisy quantum devices for system sizes that out of reach for other methods.

\begin{acknowledgments}
\paragraph{{Acknowledgements.---}}
We thank Sandra Coll-Vinent, Sirui Lu and Thomas O'Brien for insightful discussions about the sequential interferometry method, and Henrik Dreyer, Khaldoon Ghanem, K\'evin H\'emery, and Daniel Malz for valuable suggestions on applications of this algorithm.
We acknowledge the support from the German Federal Ministry of Education and Research (BMBF) through FermiQP (Grant No. 13N15890) and EQUAHUMO (Grant No. 13N16066) within the funding program quantum technologies - from basic research to market. This research is part of the Munich Quantum Valley (MQV), which is supported by the Bavarian state government with funds from the Hightech Agenda Bayern Plus. DSW has received funding from the European Union’s Horizon 2020 research and innovation programme under the Marie Sk{\l}odowska-Curie Grant Agreement No. 101023276.
The work was partially supported by the Deutsche Forschungsgemeinschaft (DFG, German Research Foundation) under Germany's Excellence Strategy -- EXC-2111 -- 390814868.
\end{acknowledgments}

\bibliography{loschmidt}

\newpage
\appendix
\setcounter{secnumdepth}{3}

\begin{center}
    {\large\textbf{Supplemental Material}}
\end{center}

\section{Errors}

\subsection{Zero-free functions}
\label{appendix:zero-free}
To apply our algorithm to the time interval $[t_1, t_2]$, we require that $\mathcal{G}(t)$ is nonzero in this region. In practice, we have to place a lower bound on the magnitude of $\mathcal{G}(t)$ to guarantee a bounded error of the algorithm. We will refer to functions that satisfy such a lower bound as zero free, following the terminology introduced in reference~\cite{Harrow2020}. Owing to its similarity to the partition function, the Loschmidt amplitude generically takes the form $\mathcal{G}(z) \sim e^{-N g(z)}$ in the limit of large $N$ for some function $g(z)$~\cite{Touchette2009, Gambassi2012, Heyl2013, Harrow2020}. This behavior motivates the following definition of a zero-free function.
\begin{definition}[Zero-free functions]
    A sequence of functions $f_N(z)$ on the complex plane is called zero free at point $z_0$ if there exist constants $c$ and $a$ such that $\log |f_N(z)| \le cN$ for all $z$ satisfying $|z - z_0| < a$.
\end{definition}

Assuming that a function is zero free allows us to bound to the derivatives of $\ln f_N(z)$ according to the following lemma.
\begin{lemma}
    Consider a sequence of holomorphic functions $f_N(z)$, which are zero free at $t_0 \in \mathbb{R}$. We let $z = t - i\beta$ with $t, \beta \in \mathbb{R}$ and $\log f_N(z) = \log r_N(z) + i\phi_N(z)$ for some analytic branch of the logarithm. Given any constant $m$, the magnitude of the partial derivatives $\partial^{m} \ln r_N(z) / \partial \beta^{m}$ and $\partial^{m} \phi_N(z) / \partial t^{m}$ at $z = t_0$ is bounded from above by $\varO(N)$.
    \label{lemma:scaling}
\end{lemma}

\begin{proof}
To prove the lemma for $\partial^{m} \ln r_N(z) / \partial \beta^{m}$, we observe that
\begin{equation}
    \frac{\dInt^m}{\dInt z^m} \ln f_N(z) = \frac{\partial^m}{\partial ( -i \beta)^m} \ln r_N(z) + i \frac{\partial^m}{\partial ( -i \beta)^m} \phi_N(z),
\end{equation}
which implies
\begin{equation}
    \left| \frac{\partial^m}{\partial \beta^m} \ln r_N(t_0 - i \beta) \right|_{\beta = 0} \leq \left| \frac{\dInt^m}{\dInt z^m} \ln f_N(z) \right|_{z = t_0}.
    \label{eq:derivative}
\end{equation}
The right-hand side can be expressed using Cauchy's integral formula as
\begin{eqnarray}
    \begin{aligned}
        \left| \frac{\dInt^m}{\dInt z^m} \ln f_N(z) \right|_{z = t_0} & = \frac{m!}{2\pi} \left| \oint_{|\zeta| = a/2} \frac{\ln f_N(\zeta + t_0)}{\zeta^{m+1}} \dInt \zeta \right|.
    \end{aligned}
    \label{eq:cauchy_int}
\end{eqnarray}

It remains to bound the magnitude of $\ln f_N(\zeta)$ on the circle $| \zeta - t_0| = a/2$. However, the magnitude depends on the choice of the branch of the logarithm. To overcome this issue, we consider the function $\tilde{f}_N(z) = e^{- i \phi_N(t_0)} f_N(z)$. We can bound the Cauchy integral \eqref{eq:cauchy_int} in terms of $\ln \tilde{f}_N(z)$ instead of $\ln f_N(z)$ because the derivatives of the two functions are equal.

To complete the bound, we use the Schwarz integral formula to express
\begin{align}
    \ln \tilde{f}_N(z + t_0) &= \frac{1}{2 \pi i} \oint_{|\zeta|=a} \frac{\dInt \zeta}{\zeta} \, \frac{\zeta + z}{\zeta - z} \, \re \ln \tilde{f}_N(\zeta + t_0) \nonumber\\
    & \hspace{.5cm} + i \, \im \ln \tilde{f}_N(t_0)
\end{align}
for any $|z| < a$. Since $\im \ln \tilde{f}_N(t_0) = 0$, this yields
\begin{equation}
    |\ln \tilde{f}_N(z + t_0)| \leq \frac{a + |z|}{a - |z|} \max_{|\zeta| = a} \left| \re \ln f_N(\zeta + t_0) \right|
\end{equation}
By assumption, we have
\begin{equation}
    \max_{|\zeta| = a} \left| \re \ln f_N(\zeta + t_0) \right| \leq c N.
\end{equation}
By substituting these results into \eqref{eq:derivative} and \eqref{eq:cauchy_int}, we obtain
\begin{align}
    \left| \frac{\partial^m}{\partial \beta^m} \ln r_N(t_0 - i \beta) \right|_{\beta = 0} &\leq \frac{m! c N}{2 \pi} \oint_{|\zeta| = a/2} \frac{|\dInt \zeta|}{|\zeta|^{m+1}} \frac{a + |\zeta|}{a - |\zeta|} \nonumber \\
    &\leq 3 \times 2^m m! \frac{c}{a^m} N
\end{align}
For any constant value of $m$, the derivative of $\ln r_N$ is thus bounded by a quantity $\mathcal{O}(N)$.

The proof for $\partial^{m} \phi_N(z) / \partial t^{m}$ follows the same steps, starting from $| \partial^{m} \phi_N(t) / \partial t^{m}|_{t=t_0} \leq | \dInt^m / \dInt z^m \ln f_N(z) |_{z = t_0}$ in place of Eq.~\eqref{eq:derivative}.
\end{proof}

For simplicity, we omit the subscript $N$ below.

\subsection{Error in numerical derivatives and integration}
As a first illustration of the implications of the zero-free condition, let us deal with errors arising from numerical derivatives and integration. For the symmetric finite difference approximation to the derivative
\begin{eqnarray}
    \pd{\beta} \ln r(z)_{|\beta = 0} \approx \frac{\ln r(t - i h) - \ln r(t + i h)}{2h},
    \label{eq:der_approx_2}
\end{eqnarray}
the error is given by~\cite{Quarteroni2006}
\begin{eqnarray} \label{eq:diff_err}
    \Delta{\phi}_\mathrm{D} =  \varO(h^2)\cdot\left| \frac{\partial^3}{\partial\beta^3}\left[\ln r(z)\right]_{|\beta = 0} \right| = \varO(N h^2),
\end{eqnarray}
where the $N$ dependence directly comes from Lemma~\ref{lemma:scaling}, assuming that $\mathcal{G}(z)$ is zero free. This error will be multiplied by $t$ in our algorithm, following the integration in Eq.~(4) of the main text. We note that the error has the same scaling as the error $\Delta \phi_\mathrm{ITE}$ due to the approximate imaginary-time evolution. Therefore, a higher-order finite difference approximation will not improve the asymptotic scaling of the total error.

To compute the phase up to time $t$, we need to carry out the integral in Eq.~(4) of the main text. In practice, we can only estimate the rate of change of the phase at a discrete set of points such that the integration is necessarily approximate. We use the Newton-Cotes formula of degree $n$ with a fixed spacing $\tau$ between the samples, as is natural for a real-time Trotter step $\tau$. The corresponding numerical integration error is~\cite{Quarteroni2006}
\begin{eqnarray}
    \begin{aligned}
        \Delta{\phi}_\mathrm{I} & = \varO(t \tau^{s}) \max_{0\le t^{\prime} \le t}\left|\frac{\partial^s}{\partial t^s}\phi(t^{\prime})\right|  = \varO(Nt\tau^{s}),
    \end{aligned}
\end{eqnarray}
where $s$ equals $n+2$ rounded down to the closest even number ($n=1$, $s=2$ for the trapezoidal rule; $n=2$, $s=4$ for Simpson's rule). The dependence on $N$ again follows from Lemma~\ref{lemma:scaling}. 
The Trotter error in the real-time evolution, $\Delta \phi_\mathrm{RTE}$, typically dominates over the error in the numerical integration, as the coefficient $s = 4$ for Simpson's rule exceeds the order $p$ of practical low-order Trotter expansions.

\subsection{Trotter errors}
\label{appendix:trotter}
Trotter errors lead to an error in the estimation of the quantities
\begin{eqnarray}
    r(t\pm ih)^2 = \left| \braket{\psi^{\prime} | e^{-iHt} e^{\pm h H} |\psi} \right|^2.
\end{eqnarray}
The measured probabilities will have an \emph{additive} Trotter error that scales as $\varO(Nh^2)$ for imaginary-time evolution and as $\varO(Nt\tau^{p})$ for real-time evolution~\cite{Childs2021}. However, to bound the error on $\partial \ln r / \partial \beta$, it is necessary to control the \emph{multiplicative} error in $r$. This is challenging because $r$ may be exponentially small in the system size. Nevertheless, we show in this section that the above error scalings also apply to the  numerical derivative in Eq.~(\ref{eq:der_approx_2}) assuming that particular functions are zero free. We require that $\mathcal{G}(z)$, its Trotterized version $\mathcal{G}_{\mathrm{ITE}}(z)$ defined in Eq.~\eqref{eq:gite}, and the function $\mathcal{F}_{t_0}(z, w)$ in Eq.~(\ref{eq:fw}) are both zero free. We highlight that if these assumptions are violated, it may nevertheless be possible to compute the phase using the correction method described in Appendix~\ref{appendix:zero}.

When expanded as the Taylor series in $h$, the finite difference in Eq.~(\ref{eq:der_approx_2}) has the form
\begin{eqnarray}
    \begin{aligned}[b]
        A(t, h) & = \frac{\ln r(t - i h) - \ln r(t + i h)}{2h}\\
        & = \pd{\beta} \ln r(z)_{|\beta = 0} + \frac{h^2}{6} \frac{\partial^3}{\partial \beta^3} \ln r(z)_{|\beta = 0} + \varO(h^3)\\
        & = \pd{\beta} \ln r(z)_{|\beta = 0} + \varO(Nh^2).
    \end{aligned}
    \label{eq:finite_der}
\end{eqnarray}
The $\varO(N)$ dependence in the last line comes again from applying Lemma~\ref{lemma:scaling} to $\mathcal{G}(z)$. We will us this form below.

For simplicity, we consider Trotter errors in the imaginary-time and real-time evolution separately. It is straightforward to show that the individual errors add when imaginary-time and real-time evolution are Trotterized at the same time.

\subsubsection{Imaginary-time evolution}
\label{appendix:ITE}

In this section, we provide bounds for the errors incurred by replacing the imaginary-time evolution by a product of local unitaries. We obtain different bounds for product states and for states with exponentially decaying correlations. For product states, the imaginary-time evolution can be realized by local unitaries that have the same support as the local terms appearing the Hamiltonian.  The only error in this case arises from Trotterization. For correlated states, more complicated unitaries that act on larger subsystems are required. There are additional errors that depend on the size of the subystem on which the unitaries act.

We consider a decomposition of the Hamiltonian $H$ into $\Gamma = \varO(N)$ local terms,
\begin{eqnarray}
    H = \sum_{j=1}^{\Gamma} H_j.
\end{eqnarray}
To implement the imaginary-time evolution, we wish to find unitary operators $U_j(h)$ that satisfy
\begin{equation}
    U_j(h) \ket{\psi} = \frac{1}{c_j} e^{-h H_j} \ket{\psi} ,
    \label{eq:u}
\end{equation}
where $c_j = \sqrt{\braket{\psi | e^{-2hH_j} |\psi}}$ ensures normalization. Although these unitary operators are not unique, we may always choose them such that $U_j(h)$ is a smooth function of $h$ with $U_j(0) = \mathbb{I}$.

To bound the errors due to the imaginary-time evolution, it is necessary to be more explicit about the dependence $U_j(h)$ on $h$. To this end, we let
\begin{eqnarray}
    B_j = i \frac{\mathrm{d}}{\mathrm{d}h}U(h)_{| h = 0}
\end{eqnarray}
such that $V_j(h)=e^{-iB_j h}$ equals $U_j(h)$ to first order in $h$.
By differentiating Eq.~\eqref{eq:u}, it follows that the Hermitian operator $B_j$ satisfies
\begin{equation}
    B_j \ket{\psi} = i (H_j - \langle H_j \rangle ) \ket{\psi}.
    \label{eq:u}
\end{equation}
This equation always has a (non-unique) solution as it specifies a single column of the matrix representation of $B_j$ in an orthonormal basis that includes $\ket{\psi}$ as a basis vector.

For a particular choice of $B_j$, we replace the imaginary-time evolution by $\prod_{j=1}^{\Gamma} c_j V_j(h)$. The corresponding Loschmidt amplitude is given by
\begin{equation}
    \mathcal{G}_{\mathrm{ITE}}(t - i h) = \braket{\psi' | e^{- i H t} \prod_{j=1}^\Gamma c_j V_j(h) | \psi}.
    \label{eq:gite}
\end{equation}
The unitary operations can be directly realized experimentally and it is straightforward to keep track of the constants $c_j$ classically. We note that the ordering of the operators can be chosen arbitrarily as it does not affect the result to lowest order in $h$. Indeed, one can verify by direct calculation that
\begin{eqnarray}
    \pd{\beta} \ln r(z)_{|\beta = 0} = \pd{\beta} \ln r_{\mathrm{ITE}}(z)_{|\beta = 0},
\end{eqnarray}
where $r_\mathrm{ITE}(z) = | \mathcal{G}_{\mathrm{ITE}}(z) |$.

To bound the error of the finite-difference $A_{\mathrm{ITE}}(t,h) = [\ln r_\mathrm{ITE}(t-ih) - \ln r_\mathrm{ITE}(t+ih)] / 2h$, it is necessary to derive a bound on the higher derivatives of $\ln r_\mathrm{ITE}(z)$. However, $\mathcal{G}_\mathrm{ITE}(z)$ is not holomorphic due to the separate dependence on the real and imaginary parts of $z$. To be able to apply Lemma~\ref{lemma:scaling}, we define
\begin{equation}
    f_t(z) = \braket{\psi' | e^{- i H t} \prod_{j = 1}^\Gamma  \sqrt{\braket{\psi | e^{- 2 z H_j} | \psi}} {e^{i z B_j}} | \psi},
\end{equation}
which is holomorphic in $z$ and satisfies
\begin{equation}
    \frac{\partial^m}{\partial \beta^m} \ln r_\mathrm{ITE}(t - i \beta) = \frac{\partial^m}{\partial \beta^m} \ln | f_t(\beta) | .
\end{equation}
If we assume that $f_t(z)$ is zero free at $z=0$, then Lemma~\ref{lemma:scaling} implies
\begin{eqnarray}
    \frac{h^2}{6}\frac{\partial^3}{\partial \beta^3} \ln r_{\mathrm{ITE}}(z)_{|\beta = 0} = \varO(Nh^2),
\end{eqnarray}
from which it follows that
\begin{eqnarray}
    |A_{\mathrm{ITE}}(t,h) - A(t,h)| = \varO(Nh^2).
    \label{eq:error_ITE1}
\end{eqnarray}

When the initial state $\ket{\psi}$ is a product state, the support of $U_j(h)$ can be restricted to the sites on which $H_j$ acts. Therefore, $B_j$ can be computed efficiently on a classical computer.
For more general states, it becomes impractical to solve Eq.~\eqref{eq:u} exactly. Nevertheless, if the correlations of the state $\ket{\psi}$ decay exponentially, then $U_j(h)$ can be approximated by a local unitary operator $\tilde U_j(h)$ such that
\begin{equation}
    \left\lVert \prod_{j=1}^\Gamma \tilde{U}_j(h) \ket{\psi} -  \prod_{j=1}^\Gamma \frac{e^{- h H_j}}{c_j} \ket{\psi} \right\rVert \leq \eta .
\end{equation}
According to Theorem 1 in the supplemental information of the paper by Motta et al.~\cite{Motta2020}, it is sufficient for $\tilde U_j(h)$ to act on
\begin{eqnarray}
    N_q = \varO\left(\left[\xi  \ln (N / \eta) \right]^d\right)
\end{eqnarray}
qubits, where $d$ is the spatial dimension and $\xi$ is the maximum correlation length of the sequence of states $\prod_{j=1}^k \tilde{U}_j(h) \ket{\psi}$ for $0 \leq k \leq \Gamma$. Since the value of $h$ is small, we expect that the correlation length of $\ket{\psi}$ is typically a good estimate of $\xi$.

Replacing the unitary operators $U_j(h)$ by their local counterparts $\tilde{U}_j(h)$ results in a modified Loschmidt echo $\tilde{r}_\mathrm{ITE}(z)$ whose error is bounded by
\begin{equation}
    | \tilde{r}_\mathrm{ITE}(z) - r_\mathrm{ITE}(z) | \leq \eta .
\end{equation}
This leads to an error in the finite difference given by
\begin{equation}
    |\tilde{A}_\mathrm{ITE}(t, h) - A_\mathrm{ITE}(t, h)| = \mathcal{O}( \eta / h r(t) ) ,
\end{equation}
which is to be added to the discretization error of Eq.~(\ref{eq:error_ITE1}). For both errors to be bounded by $\epsilon$, we let $h = \mathcal{O}(\epsilon^{1/2} / N^{1/2})$ and $\eta = \mathcal{O}(r(t) \epsilon^{1/2} / N^{1/2})$.

To obtain $\tilde{U}_j$, tomography on $N_q$ sites is required, which incurs a computational cost that scales as $\exp(\varO(N_q))$. The cost is polynomial in the system size $N$ and the inverse of desired error $1/\epsilon$ for one-dimensional systems and quasi-polynomial in $d \ge 2$ dimensions, provided $r(t)$ is bounded from below by an inverse polynomial in $N$.

\subsubsection{Real-time evolution}
\label{appendix:RTE}
For the real-time evolution, we consider
\begin{equation}
    \label{eq:grte}
    \mathcal{G}_{\mathrm{RTE}}(t \pm i h) = \braket{\psi' | U_\mathrm{RTE}(t) e^{\pm h H} | \psi},
\end{equation}
where $U_{\mathrm{RTE}}(t)=U_{p}(t / D)^{D}$ is $p$-th order Trotter decomposition of the real-time evolution with a total number of $D$ Trotter steps. We define the multiplicative error operator $M(t)$ by
\begin{eqnarray}
    U_{\mathrm{RTE}}(t) = [1 + M(t)] e^{- i H t}.
\end{eqnarray}
The error operator is bounded in operator norm by $\Vert M(t) \Vert = \varO(Nt\tau^{p})$~\cite{Childs2021}, where $\tau = t / D$~\cite{Childs2021}.

We define $A_{\mathrm{RTE}}(t,h)$ as the approximation to the finite difference $A(t, h)$ with $r_\mathrm{RTE}(z) = |\mathcal{G}_{\mathrm{RTE}}(z)|$ in place of $r(z)$. The difference $A_{\mathrm{RTE}}(t,h) - A(t,h)$ is given by
\begin{eqnarray}
    \frac{1}{2}\pd{\beta} \ln g(z)_{|\beta=0} + \frac{h^2}{12} \frac{\partial^3}{\partial \beta^3} \ln g(z)_{|\beta = 0} + \varO(h^3),
    \label{eq:RTE_diff}
\end{eqnarray}
where
\begin{eqnarray}
    \label{eq:g}
    \begin{aligned}
        g(z) = \frac{r_{\mathrm{RTE}}(z)^2}{r(z)^2} = \left|1 + \frac{\braket{ \psi' | M(t) | \psi(z) }}{\braket{ \psi' | \psi(z) }}\right|^2,
    \end{aligned}
\end{eqnarray}
with $\ket{\psi(z)} = e^{-i H z} \ket{\psi}$. In contrast to the imaginary-time evolution, the first-order derivative does not cancel. It contributes the leading-order error, which we analyze in what follows.

We would like to again apply Lemma~\ref{lemma:scaling}, which leads us to define
\begin{eqnarray}
    \mathcal{F}_{t_0}(z, w) = \braket{ \psi' |  e^{M(t_0)w} | \psi(z) },
    \label{eq:fw}
\end{eqnarray}
where $w$ is an independent complex variable from $z$. With this definition, $g(z) = |1 + \frac{\mathrm{d}}{\mathrm{d}w} (\ln\mathcal{F}_{t}) (z, 0)|^2$. Let us assume that $\mathcal{F}_{t_0}(z, w)$ is zero-free for both $z$ and $w$ at $z = t_0$ and $w=0$, i.e., $\ln |\mathcal{F}_{t_0} (z,w) | \le c N$ for all $z, w$ such that $|z-t_0| < a_z$ and $|w| < a_w$. Since $\Vert M(t_0) \Vert = \varO(Nt_0 \tau^{p})$, it is natural to choose $a_w = \varO(1 / (t_0 \tau^{p}))$, while keeping $a_z = \varO(1)$ and  $c = \varO(1)$. According to Lemma~\ref{lemma:scaling}, we get $| \frac{\mathrm{d}}{\mathrm{d}w} (\ln\mathcal{F}_{t_0}) (z, 0) | = \varO(N / a_w) = \varO(Nt_0 \tau^{p})$. Assuming that $N t_0 \tau^p$ is small, it follows that $g(z) -  1 = \varO(Nt_0\tau^{p})$ to leading order for $|z - t_0| < a_z$.

Finally, we have to confirm that the partial derivative in Eq.~(\ref{eq:RTE_diff}) does not change the system-size dependence. The leading term at $t = t_0$ is
\begin{equation}
    \frac{1}{2} \frac{\partial}{\partial \beta} g(z)_{| z = t_0} \approx \re~\pd{\beta}\frac{\mathrm{d}}{\mathrm{d}w} \ln \mathcal{F}_{t_0}(z,w)_{|z = t_0, w=0}.
    \label{eq:fw2}
\end{equation}
Similar to the proof of Lemma~\ref{lemma:scaling}, Eq.~(\ref{eq:fw2}) can be bounded as 
\begin{eqnarray}
    \begin{aligned}
        & \left| \re~\pd{\beta}\frac{\mathrm{d}}{\mathrm{d}w} \ln \mathcal{F}_{t_0}(z,w)_{|z = t_0, w=0} \right| \\
        & \le \left| \frac{\mathrm{d}}{\mathrm{d} z}\frac{\mathrm{d}}{\mathrm{d}w} \ln \mathcal{F}_{t_0}(z,w)_{|z = t_0, w=0} \right| \\
        & = \frac{1}{(2\pi)^2} \left| \oint_{|\zeta| = a_z/2} \oint_{|\eta| = a_w/2} \frac{\ln \mathcal{F}_{t_0}(\zeta + t_0, \eta)}{\zeta^{2} \eta^2} \dInt \zeta \dInt \eta \right| \\
        & \le \frac{c N}{(2 \pi)^2} \oint_{|\zeta| = a_z/2} \oint_{|\eta| = a_w/2} \frac{|\dInt \zeta| |\dInt \eta|}{|\zeta\eta|^{2}}  \cdot \frac{a_z + |\zeta|}{a_z - |\zeta|} \le \frac{12 cN}{a_z a_w}.
    \end{aligned}
\end{eqnarray}
Thus Eq.~(\ref{eq:fw2}) still scales as $c N / a_z a_w = \varO(N t_0 \tau^p)$ and we find that, to leading order,
\begin{eqnarray}
    A_{\mathrm{RTE}}(t,h) - A(t, h) = \varO(Nt\tau^{p}).
\end{eqnarray}

\section{Resource cost}
\label{appendix:existing_method}

\subsection{Estimating the magnitude $r$}
An individual circuit to measure the magnitude $r$ is required for our method and for sequential interferometry. It consists of Trotterized real-time evolution, which incurs a standard $p$-th order Trotter error~\cite{Childs2021} and statistical error
\begin{eqnarray}
    \Delta r = \varO\left(Nt \tau^{p} + \frac{1}{\sqrt{M}}\right),
\end{eqnarray}
where $M$ is the number of measurements. To bound this error within $\epsilon$, we require a circuit depth $D_r = \varO(t / \tau) = \varO(  t^{1+\frac{1}{p}} N^{\frac{1}{p}} / \epsilon^{\frac{1}{p}} )$ and the number of measurements $M_r = \varO \left( 1 /\epsilon^2 \right)$.

\subsection{Hadamard test}

\begin{figure}[H]
    \centering
    \includegraphics[width=.45\textwidth,trim={0 1.1cm 0 1.1cm},clip]{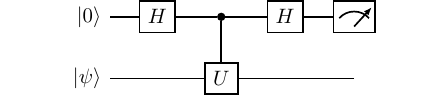}
    \caption{Quantum circuit diagram of the Hadamard test.}
    \label{fig:hadamard}
\end{figure}

The Hadamard test is the standard method to compute the real and imaginary part of $\braket{\psi | U | \psi}$ for a given unitary $U$ and initial state $\ket{\psi}$. The circuit of the Hadamard test is shown in Fig.~\ref{fig:hadamard}. One first applies a Hadamard gate to an ancillary qubit. It is followed by a controlled unitary c-$U$ acting on the prepared state $\ket{\psi}$ conditioned on the first ancillary qubit. Finally, apply the Hadamard gate again to the ancilla and measure this qubit. The probability of measuring $0$ is
\begin{eqnarray}
    \frac{1 + \re~\mathcal{G}(t)}{2} = \frac{1 + r(t)\cos\phi(t)}{2}.
\end{eqnarray}
To measure the imaginary part, we modify the circuit by adding a $S^{\dagger} = \begin{pmatrix}
    1 & 0 \\ 0 & -i
\end{pmatrix}$ gate after the first Hadamard gate.

It is hence possible to infer $\mathcal{G}(t)$ directly from the measured probabilities. Shot noise gives rise to the statistical error
\begin{eqnarray}
    \Delta{\mathcal{G}}_{\mathrm{S}}(t) = \varO\left(  1 / \sqrt{M} \right),
\end{eqnarray}
where $M$ is the number of measurements. Compared to our proposed method, there is no need of integration, so the Trotter error does not have an additional $t$ dependence:
\begin{eqnarray}
    \Delta{\mathcal{G}}_{\mathrm{RTE}} (t) = \varO(N t\tau^p).
\end{eqnarray}
The total error is hence given by
\begin{eqnarray}    \begin{aligned}
        \Delta{\mathcal{G}}(N, t, p, \tau, M) = \varO\left(Nt\tau^{p} + \frac{1}{\sqrt{M}} \right).
    \end{aligned}
\end{eqnarray}
Note that this only gives the phase of a single time $t$ whereas our algorithm measures the phase on an interval $[0, t]$.

To bound the error $\Delta{\mathcal{G}} < \epsilon$, we need a Trotter step $\tau = \varO ( ( \epsilon / N t )^\frac{1}{p} )$ and a number of measurements $M = \varO ( 1  / \epsilon^2)$. If given access to global controlled unitary evolution, the corresponding circuit depth is $D = \varO(t / \tau) = \varO( t^{1+\frac{1}{p}} N^{\frac{1}{p}} / \epsilon^{\frac{1}{p}} )$. However, on current devices usually only local gates are available. In this case there are two choices to implement this controlled unitary:
\begin{enumerate}
    \item After Trotterization, swap the qubits of each local unitary next to the control qubit to perform local controlled evolution. The swapping process will increase the depth of circuit by $\varO(N^{1+1/d})$ times, and thus $D = \varO(t / \tau) = \varO( t^{1+\frac{1}{p}} N^{1 + \frac{1}{p} + \frac{1}{d}} / \epsilon^{\frac{1}{p}} )$.
    
    \item Distribute the control qubit onto a Greenber\-ger–Horne–Zeilinger (GHZ) state of $\varO(N)$ ancillary qubits~\cite{Monz2011, Song2017, Friis2018, Omran2019, Wei2020}. This is, however, challenging on current devices due to the large demand of ancillary qubits, and the consequently greater rate of decoherence~\cite{Monz2011, Ozaeta2019} compared to methods that requires only one or no ancillary qubit.
\end{enumerate}

\subsection{Sequential interferometry}
Suppose we know $\braket{\psi_i| e^{-iHt} | \psi_i}$ for some state $\ket{\psi_i}$. It is then possible to compute the Loschmidt amplitude $\braket{\psi_j| e^{-iHt} | \psi_j}$ by preparing superpositions of $\ket{\psi_i}$ and $\ket{\psi_j}$ with a tunable phase difference $\theta$~\cite{Lu2021}. We denote $V_{ij}(\theta)$ a unitary that prepares such a state, i.e.,
\begin{equation}
    V_{ij}(\theta)\ket{\psi_i} = \frac{1}{\sqrt{2}} \left(\ket{\psi_i} + e^{i\theta} \ket{\psi_j} \right).
    \label{eq:V_ij}
\end{equation}
For simplicity, we assume $\ket{\psi_j}$ is orthogonal to $\ket{\psi_i}$, although the procedure can be readily generalized to non-orthogonal states.

Let us introduce the notation
\begin{eqnarray}
    \braket{\psi_x| e^{-iHt} | \psi_y} = r_{xy} e^{i\phi_{xy}}.
\end{eqnarray}
We have access to all the $r$'s from direct measurement of probabilities and we know $\phi_{ii}$ by assumption. The goal is to determine $\phi_{jj}$. To do so, we can follow a two-step procedure:

\begin{enumerate}
    \item Determine the phase of the cross term, $\phi_{ij} = -\phi_{ji}$. We need to measure
          \begin{eqnarray}
              \begin{aligned}
                    & \left|\braket{\psi_i | e^{-iHt} V_{ij}(\theta)| \psi_i}\right|^2                                 \\
                  = & \frac{1}{2} \left[ r_{ii}^2 +r_{ij}^2 + 2r_{ii} r_{ij} \cos(\phi_{ij}-\phi_{ii}+\theta) \right].
              \end{aligned}
          \end{eqnarray}
          In the case when $r_{ij}\neq 0$, only two different values of $\theta$ are needed to determine $\phi_{ij}-\phi_{ii}$ and thus $\phi_{ij}$.
    \item Now measure
          \begin{eqnarray}
              \begin{aligned}
                    & \left|\braket{\psi_j | e^{-iHt} V_{ij}(\theta)| \psi_i}\right|^2                               \\
                  = & \frac{1}{2} \left[r_{jj}^2 +r_{ij}^2 + 2r_{jj} r_{ij} \cos(\phi_{jj}-\phi_{ij}+\theta)\right].
              \end{aligned}
          \end{eqnarray}
          Again, two different value of $\theta$ are sufficient to determine $\phi_{jj}$.
\end{enumerate}

To prepare a global superposition state as in Eq.~(\ref{eq:V_ij}) can be difficult. If the target state and the state for which the phase is known are both product states, it is possible to repeat the interferometry by flipping one or a few spins each time. Such a sequential approach requires only local gates. For a single call of this algorithm, the error in the phase difference $\phi_{jj} - \phi_{ii}$ is bounded by the error of extracting the phases in the terms $r_{ii} r_{ij} \cos(\phi_{ij}-\phi_{ii}+\theta)$ and $r_{jj} r_{ij} \cos(\phi_{jj}-\phi_{ij}+\theta)$. Hence, for a single step,
\begin{eqnarray}
    \Delta \phi = \varO \left( Nt \tau^{p} +  \frac{1}{ \tilde{r}_{ij}^2 \sqrt{M} } \right),
\end{eqnarray}
where
\begin{eqnarray}
    \frac{1}{\tilde{r}_{ij}^2} = \frac{1}{r_{ii} r_{ij}} + \frac{1}{r_{ij} r_{jj}}.
\end{eqnarray}

For the sequential version, on average $\varO(N)$ calls are required for an arbitrary product state. Thus, the total phase error $\Delta \phi_N$ will be amplified as
\begin{eqnarray}
    \Delta \phi_N = \varO \left( N^2 t \tau^{p} +  \frac{N\tilde{I} }{\sqrt{M}} \right),
    \label{eq:amplified}
\end{eqnarray}
where
\begin{eqnarray}
    \tilde{I} = \frac{1}{\lambda N} \sum_{i=0}^{\lambda N-1} \frac{1}{\tilde{r}_{i,i+1}^2},
\end{eqnarray}
for a total number of $\lambda N$ steps in the sequence. We assume that $\lambda = \varO(1)$ and omit it in Eq.~\eqref{eq:amplified}. Therefore, to control $r \Delta \phi_N$ within $\epsilon$, one arrives at $D = \varO( r^{\frac{1}{p}} t^{1+\frac{1}{p}} N^{\frac{2}{p}}  / \epsilon^{\frac{1}{p}})$ and $M = \varO(\tilde{I}^2 r^2 N^2 / \epsilon^2)$. 

We note that the above procedure fails when $r_{ij} \approx 0$. In this case, the following quantity is simplified as
\begin{eqnarray}
    \begin{aligned}
                & \left| \braket{\psi_i | V_{ij}^{\dagger}(0) e^{-iHt} V_{ij}(\theta)| \psi_i}\right|^2         \\
        \approx & \frac{1}{4} \left[r_{ii}^2 +r_{jj}^2 + 2r_{ii} r_{jj} \cos(\phi_{jj}-\phi_{ii}+\theta)\right]
    \end{aligned}
\end{eqnarray}
so that we can measure it for different values of $\theta$ and directly obtain the phase difference between $\phi_{ii}$ and $\phi_{jj}$. The definition of $\tilde r_{ij}$ will be modified accordingly.

\section{Zeros in Loschmidt amplitude}
\label{appendix:zero}

\begin{figure}
    \centering
    \includegraphics[width=.49\textwidth]{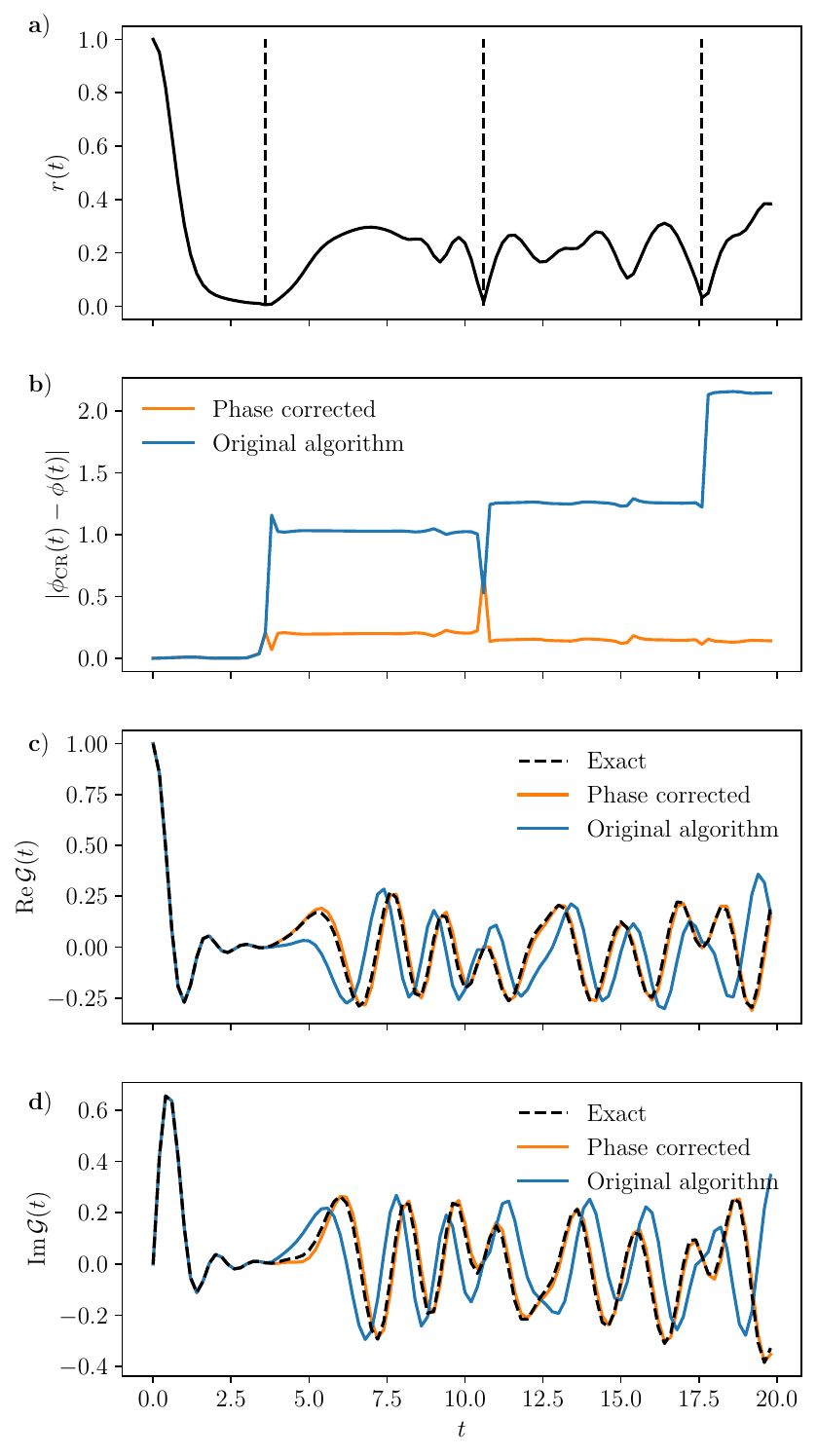}
    \caption{An example of phase correction when $\mathcal{G}(t)$ is close to 0. Here the Hamiltonian coefficients are $(J,g)=(1,1)$ and the initial state is still $\ket{\psi} = \ket{\psi'} = \ket{\uparrow\uparrow\cdots}$. The system size $N=10$. The Trotter steps are $\tau=0.2$ and $h=0.01$. \textbf{(a)}~The magnitude $r(t)$. The dashed lines indicate times at which $r(t)$ almost vanishes. \textbf{(b)}~The phase difference between the estimated and the exact value. \textbf{(c-d)}~The real and imaginary parts of $\mathcal{G}(t)$.
        Original algorithm corresponds to our algorithm described in main text. In the ``phase corrected'' results, we added a discrete phase jumpy $\pi$ as well as a phase shift $\delta$, which ensures continuity of the first derivative of $\mathcal{G}(t)$.}
    \label{fig:loschmidt}
\end{figure}

The Cauchy-Riemann equation for $\ln \mathcal{G}(z)$ holds only when $\mathcal{G}(z)\neq 0$. When $\mathcal{G}(z)$ crosses a zero at some time $t_0$, it is subject to a phase jump. In the case when we have access to arbitrary resolution and precision, the phase factor is $(-1)^{n_0}$, where $n_0$ is the smallest integer such that the $n_0$-th order derivative of $\mathcal{G}(z)$ is nonzero at $t_0$. This can be seen from the Taylor series expansion of $\mathcal{G}(z)$ at $t_0$:
\begin{eqnarray}
    \mathcal{G}(z) = \frac{1}{n_0!} \mathcal{G}^{(n_0)}(t_0) (z - t_0)^n + \varO\left((z-t_0)^{n_0+1}\right).
\end{eqnarray}
For a small $\epsilon > 0$., approximately
\begin{eqnarray}
    \mathcal{G}(t_0 - \epsilon) \approx (-1)^{n_0} \mathcal{G}(t_0 + \epsilon)
\end{eqnarray}
and thus
\begin{eqnarray}
    \lim_{t\to t_0^{+}}\phi(t) \equiv \lim_{t\to t_0^{-}}\phi(t) + n_0 \pi \pmod{2\pi}.
\end{eqnarray}

In practice, with finite time resolution (e.g., the Trotter time), the gradient of the phase becomes singular near $t_0$ (or even some point where $\mathcal{G}(t)$ is almost zero). Even if we have included the $(-1)^{n_0}$ phase jump, this singularity can cause an additional phase factor $e^{i\delta}$ when going across $t_0$. Thus our algorithm will give a numerical result $\tilde{\mathcal{G}}(z)$, where
\begin{eqnarray}
    \tilde{\mathcal{G}}(t) = \left\{
    \begin{aligned}
         & \mathcal{G}(t),            & t < t_0; \\
         & \mathcal{G}(t)e^{i\delta}, & t > t_0.
    \end{aligned}
    \right.
\end{eqnarray}

It is possible to correct the phase error $\delta$ by requiring that the $n_0$-th derivative of $\mathcal{G}$ be continuous. In particular, we can directly estimate $\delta$ from the expression
\begin{eqnarray}
    e^{i \delta} = \lim_{t\to t_0^{+}}\tilde{\mathcal{G}}^{(n_0)}(t) /  \lim_{t\to t_0^{-}}\tilde{\mathcal{G}}^{(n_0)}(t).
\end{eqnarray}
In the discrete time setting, we evaluate the limits at the closest points on either side of the zero.

An example of the algorithm in the presence of small values of $r(t)$ is shown in Fig.~\ref{fig:loschmidt}. While there are no exact zeros in the Trotterized simulation, $r(t)$ becomes very small at the dashed lines in panel (a). In the original algorithm, this leads to large phase jumps because the discretization is too coarse. As shown in panels (b)--(c), it is possible to partially correct these jumps using the method described above, where we impose continuity of the first derivative close at the small values of $r(t)$.

\section{Extension to multiple time evolution operators}
\label{appendix:multi_evol}

It is possible to extend our algorithm to multiple time evolution operators. For example, one may consider functions of the form
\begin{equation}
\mathcal{G}_A(t)=\langle \psi'| e^{-iH't} A e^{-iHt} |\psi \rangle,
\end{equation}
where the times in both evolutions are the same and $A$ can be any unitary operation. Here, $\mathcal{G}_A(t)$ is again an analytic function of $t$, so that the phase can directly be extracted by computing the derivative along the imaginary direction. The imaginary-time evolution can be simultaneously applied on both sides under the condition both $|\psi\rangle$ and $|\psi'\rangle$ are short-range correlated.

Our approach can also be generalized for $n$ evolution operators with different time variables:
\begin{eqnarray}
    \mathcal{G}(t_1, \cdots, t_n) = \braket{\psi^{\prime} | U_1 O_{1} U_{2} O_2 \cdots O_{n-1} U_n | \psi},
\end{eqnarray}
where $U_j =  e^{-iH_j t_j}$ and unitaries $O_j$ are local for $j > 1$. One can compute these quantities using the following protocol.

\begin{itemize}
    \item First switch on only $U_1$, i.e., set $t_2 = \cdots = t_n = 0$. Our algorithm works since the new intial state $O_1\cdots O_{n-1}\ket{\psi}$ still has finite correlation length.
    \item Then switch on $U_2$ as well. For each fixed $t_1$, perform our algorithm with initial state $O_2\cdots O_{n-1}\ket{\psi}$, evolution operator $e^{-iH_2 t_2}$ and final state $e^{iH_1t_1}\ket{\psi^{\prime}}$. The phase at $t_2=0$ has already been determined in the previous step.
    \item Switch on the rest of the evolution unitaries one by one.
\end{itemize}

\section{Applications}
Our algorithm is well suited to tackle a broad spectrum of problems within the domains of quantum chemistry and condensed matter physics. Below we focus on three particularly important examples: (A)~probing thermal expectation values, (B)~quantum eigenvalue estimation, and (C)~computing temporal correlation functions.

\subsection{Probing thermal properties}
Recently, a quantum algorithm was developed to probe finite-energy properties of quantum many-body systems through spectrum filtering~\cite{Lu2021}. In fact, this approach has already been implemented in a proof-of-principle experiment exploring a small instance of the Fermi-Hubbard model~\cite{Hemery2023}.
The main idea is to define the Gaussian energy filter 
\begin{eqnarray}
    P_{\delta}(E) = \exp\left[ -\frac{(H-E)^2}{2\delta^2} \right]
\end{eqnarray}
and apply it to an initial (product) state $\ket{\psi}$ whose mean energy is close to $E$. We note that it is possible to reach any constant energy density above the ground state by considering the tensor product of ``patches'', where larger patches allow for lower energies. For a given unitary observable $A$, the microcanonical ensemble average at the corresponding energy can be obtained from
\begin{eqnarray}
    A_{\delta,\psi}(E) = \frac{\braket{\psi | P_{\delta}(E) A P_{\delta}(E)| \psi}}{\braket{\psi | P_{\delta}(E)^2 | \psi}}
\end{eqnarray}
for sufficiently small values of $\delta$. Alternatively, one can filter the whole spectrum as 
\begin{eqnarray}
    A_{\delta}(E) = \frac{ \tr \left[ A P_{\delta}(E) \right] }{ \tr \left[ P_{\delta}(E) \right] }
\end{eqnarray}
and use classical Monte Carlo method to compute the trace.

The filter $P_{\delta}(E)$ can be approximated with \emph{a time series} and thus the algorithm only requires the measurement of quantities of the form
\begin{eqnarray}
    \braket{\psi | e^{iHt_1} A e^{-iHt_2} | \psi} \quad \text{or} \quad \braket{\psi | A e^{-iHt} | \psi }.
\end{eqnarray}
These quantities can efficiently be computed using the strategies outlined in section~\ref{appendix:multi_evol}.

\subsection{Quantum eigenvalue estimation}

The estimation of quantum eigenvalues, including the ground and excited state energies of a given Hamiltonian, is a crucial problem in quantum chemistry and condensed matter physics. Even though this problem can be solved using the quantum phase estimation algorithm, the resource requirements are prohibitive for current quantum computers. As a more efficient alternative, several recent proposals showed that eigenvalues can be estimated, even with the Heisenberg-limited scaling, by measuring Loschmidt amplitudes at different times followed by classical signal processing~\cite{Somma2019, Lin2022, Ding2023, Ni2023}. Our approach can replace the Hadamard test in these algorithms if the initial state is short-range correlated.

The computation of the local density of states in Fig.~3c of the main text follows the same idea, although we use a Fourier transform instead of more refined signal processing. This simple approach already yields valuable information about the energy eigenvalues, including the ground state energy, despite using a product state as the initial state. In the setting of quantum chemistry, the use of product states can be justified by the fact that the lowest energy Hartree-Fock has a substantial overlap with the ground state for a large class of molecules~\cite{Bauer2020}. In all cases, the overlap may be increased by a shallow circuit or an adiabatic evolution that maintains a short correlation length.

\subsection{Temporal correlation functions}
Temporal correlation functions play an important role in condensed matter physics, for instance in the study of transport and diffusion. They typically take the form
\begin{eqnarray}
    C_{A, B}(t) = \braket{ e^{-iHt} A e^{iHt} B },
    \label{eq:corr}
\end{eqnarray}
where $\braket{\cdots}$ represents the thermal average. Correlation functions of this type turn out to be nontrivial and hard to compute for many models even at infinite temperature~\cite{Agarwal2015, Kanasz2017, Parker2019}. At infinite temperature the thermal average can be obtained by sampling random product states, with our algorithm providing the complex expectation value for each of them (If $B$ is local, it can be absorbed into $\ket{\psi}$ in the first protocol of section~\ref{appendix:multi_evol}). Finite temperatures may be accessible by combining the approach with the above spectral filtering algorithm.

Other quantities of interest are the out-of-time-order correlators (OTOCs)~\cite{swingle2018}. OTOCs are widely used as a measure of information scrambling, and are typically defined as 
\begin{eqnarray}
    C_{\mathrm{OTOC}}(t) = - \braket{[W(t),V(0)]^2}
\end{eqnarray}
for given operators $W$ and $V$. It has recently been shown that also the imaginary part of OTOCs can contain relevant information~\cite{Sajjan2023}. Combined with the extension in Appendix~\ref{appendix:multi_evol}, measuring complex OTOCs is also possible with our algorithm.

\section{Numerical details}
\subsection{Quantum imaginary-time evolution of transverse field Ising chain}
\label{appendix:ITE_ising}
The Hamiltonian of the transverse field Ising model is given by
\begin{eqnarray}
    H = -J\sum_{i=1}^{N-1}S^z_i S^z_{i+1} + g\sum_{i=1}^{N} S^x_i.
\end{eqnarray}
When applying $e^{\pm h H}$ on a product state $\ket{\psi}$ in the computational basis $\set{\ket{\uparrow}, \ket{\downarrow}}^{\otimes N}$, it can be Trotterized as
\begin{eqnarray}
    \ket{\psi_{\pm}} \approx e^{\pm h H_2} e^{\pm h H_1} \ket{\psi} / c_{\pm},
    \label{eq:ITE_state}
\end{eqnarray}
where
\begin{eqnarray}
    \begin{aligned}
        H_1 = -J\sum_{i=1}^{N-1}S^z_i S^z_{i+1},\quad
        H_2 = g\sum_{i=1}^{N} S^x_i.
    \end{aligned}
\end{eqnarray}
$H_1$ only leads to a rescaling factor $\exp({\pm h \braket{\psi | H_1 | \psi}})$ but will not change the normalized state. For $H_2$, since
\begin{eqnarray}
    e^{a S^x} = \cosh \frac{a}{2} \id + \sinh \frac{a}{2} \sigma^x,
\end{eqnarray}
it follows that
\begin{eqnarray}
    \begin{aligned}
        e^{\pm h g S_x} \ket{\uparrow}   & = \cosh \frac{hg}{2} \ket{\uparrow} \pm \sinh \frac{hg}{2} \ket{\downarrow}, \\
        e^{\pm h g S_x} \ket{\downarrow} & = \cosh \frac{hg}{2} \ket{\downarrow} \pm \sinh \frac{hg}{2}\ket{\uparrow}.
    \end{aligned}
\end{eqnarray}
Thus for each spin, there is an additional rescaling factor
\begin{eqnarray}
    \sqrt{\cosh^2 \frac{hg}{2} + \sinh^2 \frac{hg}{2}} = \sqrt{\cosh (hg)}
\end{eqnarray}
and the spin is rotated by an angle
\begin{eqnarray}
    \theta = \arctan \tanh \frac{hg}{2}.
\end{eqnarray}
The total rescaling factor is
\begin{eqnarray}
    c_{\pm} = e^{\pm h \braket{\psi | H_1 | \psi}}\cdot \cosh (hg)^{\frac{N}{2}}.
\end{eqnarray}
For the $\ket{Z+}=\ket{\uparrow\uparrow\uparrow\cdots}$ state,
\begin{eqnarray}
    c_{\pm}^{Z+} = \left[ \exp\left(\mp\frac{hJ}{4} \cdot \frac{N-1}{N} \right) \sqrt{\cosh(hg)} \right]^{N}.
\end{eqnarray}

\subsection{Matrix product state simulations}
In Fig.~2 of the main text, the imaginary-time evolution is described in section~\ref{appendix:ITE_ising}. Given the normalized product states $\ket{\psi_{\pm}}$ from Eq.~(\ref{eq:ITE_state}), the real-time evolution is simulated with matrix product states (MPS) using the time-evolving block decimation (TEBD) algorithm of second order Trotterization. The bond dimension is chosen to be 200, which is sufficient for convergence in Fig.~2.

\subsection{Discrete Fourier transform}
\label{appendix:disc_fourier}
As we only have access to the Loschmidt amplitude up to finite maximal time $t_{\max}$ and with finite resolution $\tau$, we can only perform discrete (inverse) Fourier transform instead of the continuous one. Let $g_{k} = \mathcal{G}(k\tau)$ with $t_{\max} = K\tau$, where $\tau$ is the Trotter step. The discrete inverse Fourier transform of $g_k$ has the form
\begin{eqnarray}
    d_l = \frac{\tau}{2\pi}\sum_{k=0}^{K-1} g_ke^{i2\pi kl/K},
\end{eqnarray}
for $0\le l \le K-1$. For $d_l$ to approximate the LDOS $d(l \eta)$ at energy $l \eta$, where $\eta$ is the energy resolution, it should hold that
\begin{eqnarray}
    2\pi /K = \eta \tau,
\end{eqnarray}
which gives
\begin{eqnarray}
    \eta = 2 \pi / t_{\max}.
\end{eqnarray}
In the case when $\ket{\psi} = \ket{\psi^{\prime}}$, we additionally know that $\mathcal{G}(-t)=\mathcal{G}(t)^*$. Therefore the effective maximal evolution time is doubled and $\eta = \pi / t_{\max}$.

The range of energy is given by $ K\eta = 2\pi / \tau$. The obtained discrete LDOS is periodic. We shift the range to include the mean energy of the initial product state.

\end{document}